\documentclass[aps,floatfix,amsmath,twocolumn, pra, 10pt, notitlepage,longbibliography]{revtex4-2}

\usepackage[utf8]{inputenc}
\usepackage{mathtools,amsmath}
\usepackage{amsfonts,mathtools}
\usepackage{amssymb}
\usepackage{hyperref}
\usepackage{comment}
\usepackage{color}
\usepackage{amsthm}
\usepackage{graphicx}
\usepackage[normalem]{ulem} 
\usepackage{dsfont}
\usepackage{caption}
\usepackage{subcaption}
\usepackage{array}

\usepackage{float}
\makeatletter
\let\newfloat\newfloat@ltx
\makeatother
\usepackage{algorithm,algorithmic}

\newcommand{\tr}{{\rm Tr}}
\newcommand{\herm}{{\rm HERM}}
\newcommand{\CSPO}{{\rm CSPO}}
\newcommand{\CSPSC}{{\rm CSPSC}}
\newcommand{\STAB}{{\rm STAB}}
\newcommand{\cost}{{\rm COST}}
\newcommand{\qcost}{{\rm qCOST}}
\newcommand{\distill}{{\rm DISTILL}}
\newcommand{\LR}{{\rm LR}}
\newcommand{\J}{\mathbf{J}}

\newcommand{\CPTP}{{\rm CPTP}}

\newcommand{\Hmin}{D_{\min}^{\rm CSPO}}
\newcommand{\tJ}{\tilde{J}}

\newcommand{\bea}{\begin{eqnarray}}
\newcommand{\eea}{\end{eqnarray}}
\newcommand{\be}{\begin{equation}}
\newcommand{\ee}{\end{equation}}
\newcommand{\ba}{\begin{equation}\begin{aligned}}
\newcommand{\ea}{\end{aligned}\end{equation}}
\newcommand{\eqs}[1]{\begin{align}#1\end{align}}

\newcommand{\red}{\textcolor{red}}

\newtheorem{theorem}{Theorem}
\newtheorem{lemma}{Lemma}
\newtheorem{proposition}{Proposition}

\newtheorem{remark}{Remark}
\newtheorem*{definition}{Definition}
\newtheorem*{definition*}{Definition}

\theoremstyle{definition}

\def\1{\mathds{1}}

\def\id{\mathsf{id}}

\def\mD{\mathcal{D}}
\def\mE{\mathcal{E}}
\def\mF{\mathcal{F}}

\def\mM{\mathcal{M}}
\def\mN{\mathcal{N}}
\def\mP{\mathcal{P}}
\def\mQ{\mathcal{Q}}
\def\mS{\mathcal{S}}

\def\mU{\mathcal{U}}

\def\mc{\mathfrak{C}}
\def\md{\mathfrak{D}}
\def\ms{\mathfrak{S}}

\def\mf{\mathfrak{F}}

\def\({\left(}
\def\){\right)}
\def\[{\left[}
\def\]{\right]}

\makeatletter
\let\Hy@backout\@gobble
\makeatother

\begin{document}

\title{Quantifying dynamical magic with completely stabilizer preserving operations as free}
	
\author{Gaurav Saxena}
\email{gaurav.saxena1@ucalgary.ca}	
\affiliation{
Department of Physics and Astronomy, Institute for Quantum Science and Technology,
University of Calgary, AB, Canada T2N 1N4}

\author{Gilad Gour}
\email{gour@ucalgary.ca}
\affiliation{
Department of Mathematics and Statistics, Institute for Quantum Science and Technology,
University of Calgary, AB, Canada T2N 1N4}

\date{\today}
	
\begin{abstract}
In this paper, we extend the resource theory of magic to the channel case by considering completely stabilizer preserving operations (CSPOs) as free.
We introduce and characterize the set of CSPO preserving and completely CSPO preserving superchannels.
We quantify the magic of quantum channels by extending the generalized robustness and the min relative entropy of magic from the state to the channel domain and show that they bound the single-shot dynamical magic cost and distillation.
We also provide analytical conditions for qubit interconversion under CSPOs and show that it is a linear programming feasibility problem and hence can be efficiently solved.
Lastly, we give a classical simulation algorithm whose runtime is related to the generalized robustness of magic for channels. 
Our algorithm depends on some pre-defined precision, and if there is no bound on the desired precision then it achieves a constant runtime.
\end{abstract}

\maketitle

\section{Introduction}

In recent years, several schemes have been developed to achieve fault-tolerant quantum computation, and most of them use the stabilizer formalism \cite{ Shor96, G97, J97, G09,CTV17}.
The stabilizer formalism consists of the preparation of stabilizer states, application of Clifford gates, and measurements in the computational basis.
Within this formalism, pure non-stabilizer states (popularly known as magic states) are used as a resource to promote fault-tolerant quantum computation to universal quantum computation.
This model of quantum computation is known as the magic state model of quantum computation and finding magic distillation rates and estimating classical simulation cost of quantum circuits are active areas of research in this field~\cite{BK05,CB10,   FMM+12, MEK12, J13, CW17, CTV17,HH18,CC19, L19,WWS20, G05, VFG+12, S14, HWV+14, VMG+14,PWB15,DAP+15,BG16, HC17,BDB+17,RBD+17,BFH+17, BB+19, SC19, RLC+19, TRW21, LOH22}.
While formulating optimal rates promise better distillation protocols, improved classical simulations help benchmark the computational speedups offered by quantum computers~\cite{R05,BH12, ACB12,CAB12,E13,HHP+17,HC17, BFH+17,TBG17, KT19, HG19, SC19,  AAB+19,SRP+20}.
It follows from the Gottesman-Knill theorem that it is possible to efficiently simulate any stabilizer circuit on a classical computer, hence rendering stabilizer states and operations useless for universal quantum computation~ \cite{GK98, AG04}.
For this reason, this model fits the mold of quantum resource theories where all the states and operations that cannot provide any quantum advantage are treated as free~\cite{CG19,HWV+14, HC17, A19,SCG20, GAR+21, WWS19}.

Using the above criterion to define free elements, considerable work has been directed towards developing the resource theory of magic~\cite{ VFG+12,VMG+14,DAP+15, HC17, ADG+18,RLC+19,SC19,HG19, WWS19, LW20,SRP+20,FL20, HFH+21, TRW21}. In this process, two branches have emerged: one branch deals with odd d-dimensional qudits, and the other branch deals with the practically important case of multi-qubit systems. In the former case, a clear connection between quantum speedup and the negativity of the Wigner representation of the state/channel has been established~\cite{ G05, G06, G07,VFG+12, ME12,VMG+14,BSS16,DOB+17, WWS19, WWS20,KJ21}.
However, in the latter case, a discrete phase space approach cannot be cleanly applied without restricting free states to some subset of stabilizer states or excluding some Clifford operations~\cite{CGG+06,HWV+14,DAP+15, KL17, RBD+17,RBT+20, ZOR20, ROZ+21}.
Thus, to retain all stabilizer states and operations as free elements (in the multiqubit scenario), alternative approaches have been taken~\cite{BFH+17, S14,HC17,BSS16,BG16,BB+19,QWE19,SRP+20,SC19, ZOR20, HFH+21, OZR21}. 

In~\cite{HC17}, Howard and Campbell presented a scheme where all density matrices are decomposed as real linear combinations of pure stabilizer states.
Borrowing the idea from the resource theory of entanglement~\cite{VT99}, they introduced the robustness of magic which is the minimum $\ell_1$-norm of all such decompositions.
They showed that it is a resource monotone under all stabilizer operations and linked it to the runtime of a classical simulation algorithm, thus giving robustness of magic an operational meaning.
Using robustness of magic, they also formulated lower bounds on the cost of synthesizing magic gates.
Taking this approach forward, Seddon and Campbell enlarged the set of free operations from stabilizer operations to the set of completely stabilizer preserving operations (CSPOs) and introduced channel robustness of magic for multi-qubit channels~\cite{SC19}. 
They decomposed a channel as a linear combination of CSPOs and defined channel robustness as the minimum $\ell_1$-norm of all such decompositions. They also formulated a classical algorithm and linked its runtime with the channel robustness thus efficiently simulating a circuit consisting of CSPOs. 

Since CSPOs cannot provide any quantum advantage, we extend the resource theory of magic to the channel case by treating CSPOs as free.
We introduce two sets of free superchannels, CSPO preserving superchannels and completely CSPO preserving superchannels, to manipulate quantum channels.
Since there is no physical restriction over such sets of free superchannels, they are useful in finding fundamental limitations on the ability of a quantum channel to generate magic. 
Besides, studying such superchannels
gives us no-go results in resource interconversion tasks involving more restricted type of operations such as the set of stabilizer operations.
 
This paper is organized as follows.
In section \ref{sec:free_elements}, we define and characterize the two above-mentioned sets of free superchannels.
Then in section \ref{sec:magic_measures}, we generalize the key operational magic monotones defined for states to the channel domain, namely the generalized robustness of magic and the min relative entropy of magic.
Using these monotones, in section \ref{sec:interconversions}, we formulate single shot bounds on distilling magic from a quantum channel and the magic cost of simulating a channel under the free superchannels.
However, due to the complexity in determining whether a state is a stabilizer state or not~\cite{G06, HC17, GMC14}, finding the lower bound on distillation under completely CSPO preserving superchannels is still an open problem.
In section \ref{sec:interconversions}, we also show that interconversion among single-qubit states under CSPOs is an SDP feasibility problem and hence, can be efficiently solved.
As our last result, in section \ref{sec:classical_algo}, we provide an algorithm to classically simulate a general quantum circuit and relate the runtime of this algorithm to the generalized robustness of magic for channels.
Our algorithm is designed such that its runtime varies according to the desired precision and if there is no bound on the desired precision, the algorithm runs in constant time.

\section{Preliminaries}\label{sec:prelim}
\subsection{Notations}
In this paper, we denote all static systems using uppercase English letters and with a numerical subscript, like $A_1, B_0, R_1$, etc., and these systems will be considered as qubit (or multi-qubit) systems unless otherwise specified.
Dynamical systems will simply be denoted by English capital letters like $A, B, R$, etc., and this notation for a dynamical system, say $A$, would indicate a pair of systems such that $A = (A_0, A_1) = (A_0\to A_1)$.
The set of Hermitian matrices on system $A_1$ will be denoted by $\herm(A_1)$.
The set of density matrices on a system, say $B_1$, will be represented by $\md(B_1)$.
We will use $\psi$ and $\phi$ for pure states, and $\rho$ and $\sigma$ will be used for mixed states.
The set of all stabilizer states in system $A_1$ will be denoted by $\STAB(A_1)$.
For pure stabilizer states in system $A_1$ we will write $\phi\in \STAB(A_1)$, and notation like $\sigma\in \STAB(A_1)$ will mean a density matrix of a state taken from the stabilizer polytope which is a convex hull of pure stabilizer states.
The maximally entangled state and the unnormalized maximally entangled state on the composite systems $A_1\tilde{A}_1$ will be denoted by $\phi^+_{A_1\tilde{A}_1}$ and  $\Phi^+_{A_1\tilde{A}_1}$, respectively, where we used the tilde symbol to denote a replica of the system $A_1$.
To denote the dimension of a system, two vertical lines will be used. For example, the dimension of $B_0$ is $|B_0|$.

The set of quantum channels or completely positive and trace preserving (CPTP) maps on a dynamical system $A$ will be denoted by $\CPTP(A)$ or $\CPTP(A_0\to A_1)$.
To represent channels, calligraphic letters like $\mE, \mN,$ etc. will be used.
The notation $\mN_A$ or $\mN\in \CPTP(A_0\to A_1)$ will mean that the quantum channel $\mN$ takes an input state in $A_0$ to an output state in $A_1$.
The evolution of quantum channels is described by superchannels.
A brief discussion of superchannels is provided in Appendix~\ref{superchannel_desription} for completeness.
We will use uppercase Greek letters like $\Theta,\Omega,$ etc., to represent superchannels.
We will denote the set of superchannels by $\ms(A\to B)$ such that $\Theta\in \ms(A\to B)$ implies that the superchannel $\Theta$ takes a dynamical system in $A$ to a dynamical system in $B$.
The Choi matrix of a channel $\mN\in \CPTP(A_0\to A_1)$ is defined as $J^{\mN}_A := \mN_{A}\left(\Phi^+_{A_0\tilde{A_0}} \right)$,
where in the notation $J^{\mN}_A$, the subscript denotes the dynamical system $A=(A_0,A_1)$.
The Choi matrix of a superchannel $\Theta\in \ms(A\to B)$ will be denoted in bold as $\J^{\Theta}_{AB}$.
To denote normalized Choi matrix of a channel $\mN_A$, we will use tilde symbol over $J$ as $\tJ^{\mN}_A$.

\subsection{Stabilizer Formalism}

In this subsection, we give a brief overview of the stabilizer formalism.
For single-qubit systems, the Pauli group consists of Pauli matrices and the identity matrix, together with multiplicative factors $\pm 1,\, \pm i$. We will denote this group as $\mP_1 = (\pm 1, \pm i)\{I, X,Y,Z\}$.
For multi-qubit systems, general Pauli group on $n$-qubits consists of all n-fold tensor products of Pauli matrices (including identity), together with the multiplication factors $\pm 1, \pm i$.
We will denote the n-qubit Pauli group as $\mP_n$.
We say a pure, $n$-qubit state $|\psi\rangle$ is a stabilizer state if there exists an Abelian subgroup of the Pauli group $\mS\subset\mP_n$ such that $S|\psi\rangle = |\psi\rangle$ for all $S\in \mS$. The elements of the subgroup $\mS$ are called stabilizers of $|\psi\rangle$, and the total number of elements in $\mS$ is equal to $2^n$.
For example, the Pauli matrix $Z$ is the stabilizer of state $|0\rangle$.
For single-qubit states, there are six pure stabilizer states with the following stabilizers
\eqs{
\pm X |\pm\rangle &= |\pm \rangle\\
\pm Y|\pm i\rangle &= |\pm i\rangle\\
Z|0\rangle &= |0\rangle\\
-Z|1\rangle  &= |1\rangle\; .
}
The mixed stabilizer states of a system $A_1$ are defined as convex combination of pure stabilizer states. 
We can also define the set of stabilizer states using Clifford unitaries which are the unitaries that preserve the Pauli group under conjugation. 
Let $U$ represent an element of Clifford unitaries such that $UPU^{\dagger} \in \mP_n$ for all $P\in \mP_n$.
Then the set of stabilizer states can be represented as ${\rm conv}\{U|0\rangle\langle 0|U^{\dagger}: U \in {\rm Clifford}  \}$.
Evolution of stabilizer states under Clifford unitaries can be  efficiently tracked classically.
Further, even the measurement of Pauli operators on stabilizer states can be efficiently simulated \cite{GK98,AG04}.
A quantum circuit that comprises of Clifford unitaries, Pauli measurements, and classical randomness and conditioning, is known as a stabilizer circuit. 
The usefulness of the stabilizer formalism comes in quantum error correction and in efficiently simulating stabilizer circuits classically~\cite{GK98}.

\section{Completely stabilizer preserving operations (CSPO), CSPO preserving superchannels, and Completely CSPO preserving superchannels}\label{sec:free_elements}

The set of completely stabilizer preserving operations or CSPOs was introduced in \cite{SC19} and comprises of all the quantum operations that preserve stabilizer states in a complete sense. The set of completely stabilizer preserving operations taking system $A_0$ to system $A_1$ will be denoted by $\CSPO(A_0\to A_1)$ or $\CSPO(A)$.
Let $\mE_A\in \CPTP(A)$. Then $\mE_A$ is a completely stabilizer preserving operation if for any system $R_0$ it holds that
\eqs{
\mE_A(\rho_{R_0A_0})\in \STAB(R_0A_1)\;\; \forall \; \rho_{R_0A_0}\in \STAB(R_0A_0)\, .
}
These operations can alternatively be defined using their Choi matrices as follows
\be \label{CSPO_Choi_defn}
\mE_{A}\in \CSPO(A) \iff \dfrac{J^{\mE}_{A}}{|A_0|} \in \STAB(A)\, .
\ee
In \cite{SC19}, it was also shown that the action of CSPOs on a stabilizer state can be efficiently simulated classically.
This set is the largest known set of operations in the multi-qubit scenario that do not provide any quantum advantage and as such they  are perfect candidates for the free channels of a dynamical resource theory of magic.
To manipulate quantum channels, we choose the two natural sets of superchannels -- namely, the set of CSPO preserving superchannels and the set of completely CSPO preserving superchannels -- as the set of free superchannels in our work.
We will denote the set of CSPO preserving superchannels taking dynamical system $A$ to dynamical system $B$ by $\mf_1(A\to B)$ and the set of completely CSPO preserving superchannels taking dynamical system $A$ to dynamical system $B$ by $\mf_2(A\to B)$.
In the following two subsections we define and characterize the two sets of free superchannels.

\subsection{CSPO preserving superchannels}
\begin{definition}
Given two dynamical systems $A$ and $B$, a superchannel $\Theta\in \ms(A\to B)$  is said to be CSPO preserving superchannel if
\begin{equation}
\Theta_{A\to B}[\mN_A]\in \CSPO(B) \;\; \forall\; \mN_A\in \CSPO(A)\,.
\end{equation}
\end{definition}
Let $\{W_j\}$ be the set of stabilizer witnesses for system $B_0B_1$. 
Then, using the above definition and the set of stabilizer witnesses, we can characterize the set of CSPO preserving superchannels using their Choi matrices as follows.
The Choi matrix of a superchannel $\Theta\in \mf_1(A\to B)$ must satisfy the following conditions
\eqs{
&\J^{\Theta}_{AB}\geq 0\; ,\\
&\J^{\Theta}_{AB_0} = \J^{\Theta}_{A_0B_0}\otimes \frac{I_{A_1}}{|A_1|}\; ,\\
&\J^{\Theta}_{A_1B_0} = I_{A_1B_0}\; ,\\
&\tr\left[\J^{\Theta}_{AB}(\phi_i\otimes W_j) \right]\geq 0\quad \forall\; \phi_i \in \STAB(A_0A_1), W_j\;\label{cond:stab} .
}
In the above, the first three conditions follow from the requirement of $\Theta$ to be a superchannel \cite{G19}.
The condition in equation~\eqref{cond:stab} simply uses the fact that if a CSPO preserving superchannel takes the extreme points of the stabilizer polytope to a stabilizer state, then it will also take any convex combination of them to a stabilizer state.
However, finding all stabilizer witnesses is a hard problem, but for small dimensions, they can be found and the above characterization can be used as a set of conditions in resource interconversion tasks formulated as conic optimization problems.

\subsection{Completely CSPO preserving superchannels}
\begin{definition}

Given two dynamical systems $A$ and $B$, a superchannel $\Theta\in \ms(A\to B)$ is said to be completely CSPO preserving if
\be
 \Theta_{A\to B}[\mN_{AR}]\in \CSPO(BR)\quad\forall\; \mN\in\CSPO(AR)
\ee
\end{definition}
In other words, a superchannel is completely CSPO preserving if, for every input CSPO, the output is also CSPO, even if the superchannel acts only on a subsystem of the input channel.

\begin{theorem}
Let $\Theta \in \ms(A\to B)$. Then $\Theta\in \mf_2(A\to B)$ if and only if
\begin{equation}
    \dfrac{1}{|A_1B_0|}\J^{\Theta}_{AB} \in \STAB(AB)
\end{equation}
\end{theorem}

\begin{proof}
We first prove that if $\Theta$ is a completely CSPO preserving superchannel (i.e., belongs to $\mf_2(A\to B)$), then its normalized Choi matrix is a stabilizer state.
For the other side, we show that if a superchannel $\Theta$ is not a completely CSPO preserving superchannel, then its normalized Choi matrix is not a stabilizer state.

Let $\Theta\in \ms(A\to B)$ be a completely CSPO preserving superchannel.
By definition, a superchannel can be realized using a pre-processing channel $\mE\in \CPTP(B_0\to E_1A_0)$ and a post-processing channel $\mF\in \CPTP(E_1A_1\to B_1)$  \cite{G19}.
The normalized Choi matrix of the superchannel can be expressed in terms of these pre- and post-processing channels in the following way:
\begin{widetext}
\begin{equation}\label{normalized_Choi_superchannel}
    \dfrac{1}{|A_1B_0|}\J^{\Theta}_{AB} = \id_{A_1B_0} \otimes \left(\id_{A_0}\otimes\mF_{E_1A_1 \to B_1}\right)\circ\left( \id_{A_1}\otimes\mE_{B_0\to A_0E_1}\right) \left(\phi^+_{A_1\tilde{A_1}}\otimes \phi^+_{B_0\tilde{B_0}} \right)\; ,
\end{equation}
\end{widetext}
where $\phi^+_{A_1\tilde{A}_1}$($\phi^+_{B_0\tilde{B}_0}$) represents the maximally entangled state in the system $A_1\tilde{A}_1$($B_0\tilde{B}_0$). 
Eq.\eqref{normalized_Choi_superchannel} can be diagrammatically illustrated using Fig.\ref{Choi_superchannel}.
\begin{figure}[!h]
\centering
   \includegraphics[width=0.48\textwidth]{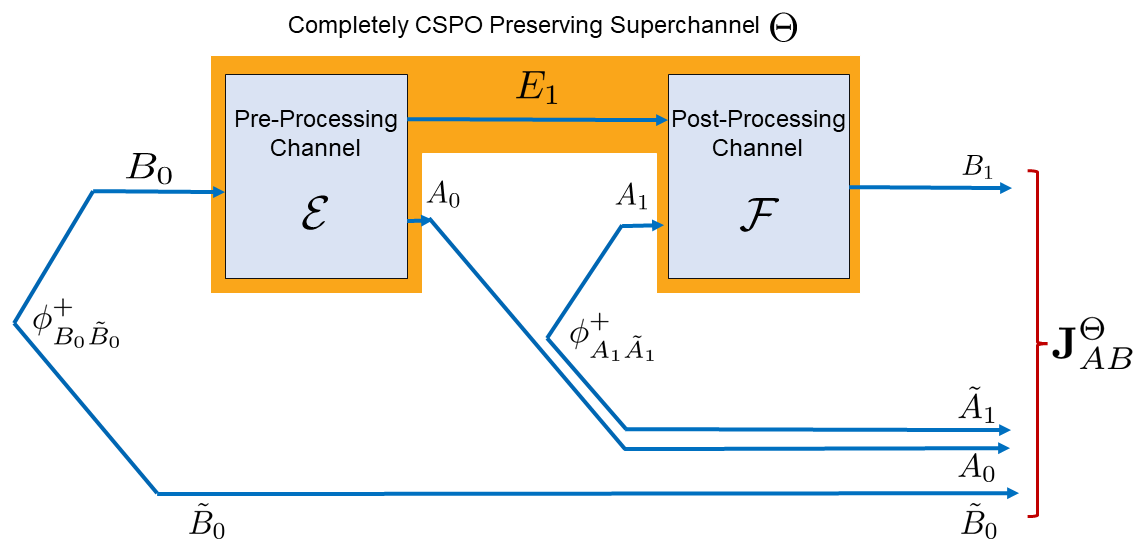}
  \caption{\linespread{1}\selectfont{\small Normalized Choi matrix of a superchannel }} 
  \label{Choi_superchannel}
\end{figure}

Define $\mN\in \CPTP(A_0\to \tilde{A_0}A_1\tilde{A_1})$ such that
\begin{equation}
    \mN(\rho_{A_0}) := \rho_{A_0} \otimes \phi^+_{A_1\tilde{A_1}}\; ,\label{useful_CSPO}
\end{equation}
for any input density matrix in $A_0$.
Note that the normalized Choi matrix of $\mN$ is a stabilizer state. 
Therefore, $\mN$ is a completely stabilizer preserving operation \cite{SC19}.
Using such a channel we can view the Choi matrix of a superchannel as shown in Fig. \ref{Choi_superchannel_as_CSPO}.

\begin{figure}[h]
\centering
   \includegraphics[width=0.48\textwidth]{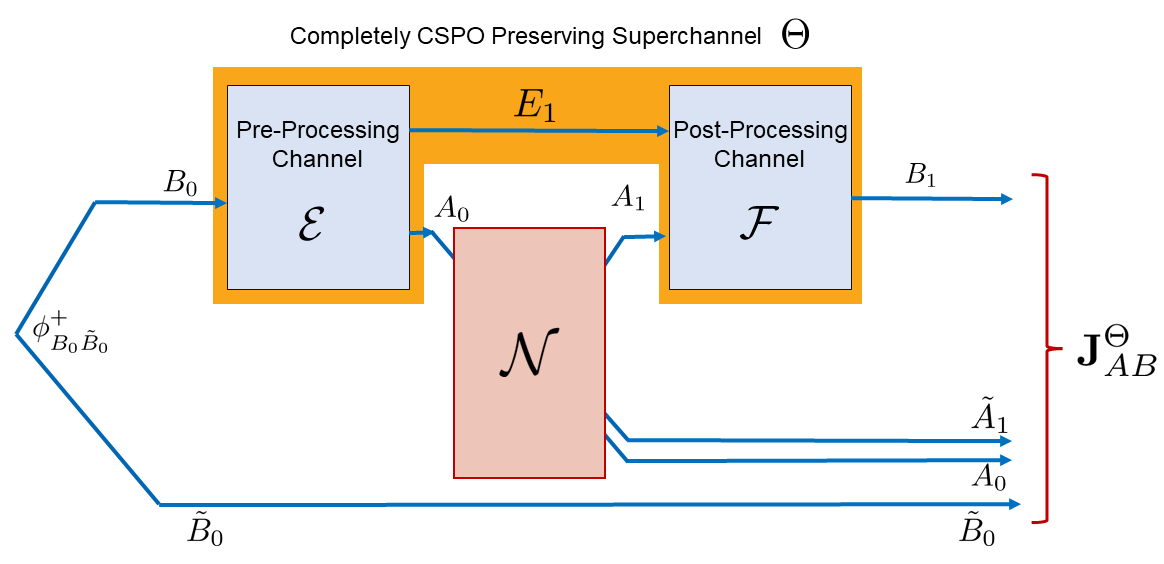}
  \caption{\linespread{1}\selectfont{\small Choi matrix of a completely CSPO preserving superchannel viewed as a CSPO}} 
  \label{Choi_superchannel_as_CSPO}
\end{figure}
Since $\Theta$ is a completely CSPO preserving superchannel, and $\mN$ is a CSPO as defined in Eq.\eqref{useful_CSPO}, the output channel $\Theta[\mN]$ is a CSPO and so, $\Theta[\mN](\phi^+_{{B_0}\tilde{B}_0})$ is a stabilizer state.

Hence, the normalized Choi matrix of a completely CSPO preserving superchannel is a stabilizer state.

For the other side of the proof, let $\Theta \in \ms(A\to B)$ be a superchannel that is not completely CSPO preserving.
Then there exists a CPTP map $\mE\in \CSPO(A_0R_0 \to A_1R_1)$ such that $\Theta_{A\to B}[\mE_{AR}]\notin \CSPO(B_0R_0\to B_1R_1)$.
Therefore, for some stabilizer witness $W_{B\tilde{R}}$, it holds that
\eqs{
    \tr\[W_{B\tilde{R}}\tr_{AR}\[\dfrac{|A_0|}{|B_0|} \J^{\1\otimes \Theta}_{ARB\tilde{R}}\( \dfrac{J^{\mE}_{AR}}{|A_0R_0|}\otimes I_{B\tilde{R}} \)  \]\] < 0\; .}
After some algebraic manipulations, the above inequality reduces to
    \eqs{ \tr\[\(\dfrac{J^{\mE}_{AR}}{|A_0R_0|} \otimes W_{B\tilde{R}} \)\dfrac{|A_0|}{|B_0|}\J^{\1\otimes \Theta}_{ARB\tilde{R}} \] < 0\;.}
Since the normalized Choi matrix of $\mE$ is a stabilizer state, the following inequality 
    \eqs{\tr\[\(|\phi\rangle_{AR}\langle\phi| \otimes W_{B\tilde{R}} \)\dfrac{1}{|A_1B_0R_0R_1|}\J^{\1\otimes \Theta}_{ARB\tilde{R}} \] < 0
}
must hold for some pure stabilizer state $|\phi\rangle_{AR}$. From \cite{SC19} we know that $(|\phi\rangle_{AR}\langle\phi| \otimes W_{B\tilde{R}})$ is a valid stabilizer witness. Hence, 
\eqs{
    \dfrac{1}{|A_1B_0R_0R_1|}\J^{\1\otimes \Theta}_{ARB\tilde{R}} \notin \STAB(ARB\tilde{R})}
which is equivalent to
    \eqs{\dfrac{1}{|A_1B_0R_0R_1|}\J^{\1}_{R\tilde{R}}\otimes \J^{\Theta}_{AB}\notin \STAB(ARB\tilde{R})}
and that implies
    \eqs{\dfrac{1}{|A_1B_0|}\J^{\Theta}_{AB}\notin \STAB(AB)\; .
}
Therefore, we can conclude that the normalized Choi matrix of a superchannel is a stabilizer state if and only if the superchannel is completely CSPO preserving.
\end{proof}

\section{Measures of magic}\label{sec:magic_measures}

In this section, we quantify the magic of quantum states and channels.
We extend the generalized robustness and the min-relative entropy of magic from the state to the channel domain \cite{SRP+20,BB+19}.
These quantifiers arise from the standard resource theoretic techniques and are related to the channel divergences which have been studied recently in detail in \cite{LY20,G19, GW19,WBH+20,FFR+20,GW21,LKD+18, CMW16}. 
Next, we formally define the geometric measure of magic for states which to the best of our knowledge has not been defined earlier.
We couldn't find any operational use of this monotone and leave it as an open problem.
Note that, we will denote the (free) robustness of magic as ${\rm R}$, the generalized robustness of magic as ${\rm R}_g$, the min relative entropy of magic of states as $D_{\min}^{\STAB}$, the hypothesis testing relative entropy of magic of states as $D_{\min}^{\epsilon,\,\STAB}$, and the min relative entropy of magic of channels as $D_{\min}^{\CSPO}$. For completeness, we have briefly discussed robustness of magic and hypothesis testing relative entropy of magic in Appendix~\ref{append:rob_of_magic} and \ref{append:hyptesting_rel_ent_magic}, respectively.

\subsection{Generalized Robustness of magic for channels}
The generalized robustness of magic for states was defined in \cite{SRP+20}.
Below we generalize it for the channel case and define the log of generalized robustness of magic for a channel $\mN_A \in \CPTP(A_0\to A_1)$ as
\eqs{
    \LR_g(\mN_A) &= \min_{\mE \in \CSPO(A_0\to A_1)} D_{\rm max} (\mN_A\|\mE_A)\label{eq:log_rob_Dmax}\\
    &= \log \min \{ \lambda : \lambda\,\mE \geq \mN;\; \mE \in \CSPO(A_0\to A_1)\} \label{eq:theta_max_defn} }
This optimization problem can be expressed in terms of Choi matrices as
\eqs{
    \LR_g(\mN_A) &=\log \min \lambda\\
    &\qquad\quad\: {\rm s.t.:}\; \lambda\, J^{\mE}_A \geq J^{\mN}_A\; ,\nonumber\\
    &\qquad\quad\qquad\; J^{\mE}_{A_0} = I_{A_0}\; ,\nonumber\\
    &\qquad\quad\qquad\; \frac{J^{\mE}_{A}}{|A_0|}\in \STAB(A_0A_1)\nonumber}
which can be simplified as
    \eqs{
     \LR_g(\mN_A) &=\log \min \frac{\tr[\omega_A]}{|A_0|}\\
    &\qquad\quad\: {\rm s.t.:}\; \omega_A \geq J^{\mN}_A\; ,\nonumber\\
    &\qquad\quad\qquad\; \omega_{A_0} = \tr[\omega_A]\frac{I_{A_0}}{|A_0|}\; ,\nonumber\\
    &\qquad\quad\qquad\; \frac{\omega_A}{\tr[\omega]}\in \STAB(A_0 A_1)\, .\nonumber
}
The dual of the above primal problem can be written as
\begin{widetext}
\ba
    \LR_g(\mN_A) &= \log \sup \tr[\alpha_A J^{\mN}_A ]\\
    &\quad\quad\quad {\rm s.t.:}\; \tr\left[\phi_i\left(\alpha_A+\beta_{A_0}\otimes I_{A_1}-\tr[\beta_{A_0}]\frac{I_A}{|A_0|}\right)\right]\leq \frac{1}{|A_0|} \;\;\forall\; \phi_i \in \STAB(A_0A_1) ,\\
    &\qquad\quad\qquad \; \alpha_A\geq 0,\, \beta_{A_0}\in \herm(A_0)\, .
\ea
\end{widetext}
Some properties of the generalized robustness of magic for channels are listed below.
\begin{enumerate}
\item \textit{Faithfulness.} $\LR_g(\mN_A) = 0  \iff \mN\in \CSPO(A_0\to A_1)$. The proof is similar to the state case.
\item \textit{Monotonicity.} $\LR_g(\Theta[\mN])\leq \LR_g(\mN)$ for any free superchannel $\Theta \in \mf_1(A\to B)$ or $\Theta \in \mf_2(A\to B)$. The proof follows from the data-processing inequality as
\ba
	\LR_g(\Theta[\mN]) &= \min_{\mF\in \CSPO(B)}D_{\max}(\Theta[\mN]\| \mF)\\
	&\leq \min_{\mE\in \CSPO(A)}D_{\max}(\Theta[\mN]\| \Theta[\mE])\\
	&\leq \min_{\mE\in \CSPO(A)}D_{\max}(\mN\| \mE)\,.
\ea
\item \textit{Sub-additivity.}  $\LR_g(\mN\otimes \mM) \leq \LR_g(\mN) + \LR_g(\mM)$. The proof easily follows from equation~\eqref{eq:log_rob_Dmax}. 
\end{enumerate}
\begin{remark}\label{remark:channel_decomp}
Eq.\eqref{eq:theta_max_defn} can be rewritten (without the {\rm log}) as 
\eqs{
{\rm R}_g(\mN_A)= \min\Big\{\lambda&\geq 1:\nonumber\\
    &\frac{\mN + (\lambda-1)\mM}{\lambda}\in \CSPO(A_0\to A_1)\,,\nonumber\\ 
    &\mM\in \CPTP(A_0\to A_1) \Big\}\, .
   }
Hence, for any $\lambda \geq {\rm R}_g(\mN_A)$, a channel $\mN_A$ can then be expressed as
\begin{equation}
    \mN_A = \lambda \mE - (\lambda - 1)\mM
\end{equation}
for some $\mE\in \CSPO(A_0\to A_1)$ and some $\mM\in \CPTP(A_0\to A_1)$.
\end{remark}
\subsection{Min-relative entropy of magic for states and channels}
Below, we present another monotone, the min relative entropy of magic.
The min-relative entropy of magic of a state $\rho$ is defined as
\begin{align}
    D_{\min}^{\STAB}(\rho) \coloneqq& \min_{\sigma\in\STAB} D_{\min}(\rho\| \sigma) \\
    =& \min_{\sigma\in\STAB}\left( -\log_2 \tr[P_{\rho} \sigma] \right)\\
    =& \,-\log_2 \max\, \tr[P_{\rho} \sigma]\\
    & \qquad\qquad\quad {\rm s.t.:}\, \sigma\in\STAB\, ,\nonumber\\
    =& \,-\log_2 \max\, \tr[P_{\rho} \phi]\\
    & \qquad\qquad\quad {\rm s.t.:}\, \phi\in \STAB\,\nonumber .
\end{align}
where $P_{\rho}$ denotes the projection onto the support of $\rho$.
Similarly, the min-relative entropy of a channel $\mN$ can be defined as
\eqs{
\Hmin(\mN_A) \coloneqq&\min_{\mE\in \CSPO(A)}D_{\min}\(\mN\|\mE\) \\
=& \min_{\mE\in \CSPO(A)}\sup_{\psi\in\md(R_0A_0)}D_{\min} \(\mN(\psi)\|\mE \(\psi\)\)\, .
}
Below we list some of the properties of the min-relative entropy of magic.
\begin{enumerate}
\item \textit{Faithfulness.} The min-relative entropy of magic is faithful for both the states and channels, i.e.,
\eqs{ 
\Hmin(\mN_A) &= 0  \iff \mN\in \CSPO(A_0\to A_1)\\
D_{\min}^{\STAB}(\rho_{A_0}) &= 0  \iff \rho\in \STAB(A_0)\, .
}
\item \textit{Monotonicity.} The min-relative entropy of magic is a monotone under CSPOs for the state case and under CSPO preserving superchannels and completely CSPO preserving superchannels for the channel case.
Thus, for any state $\rho\in \md(A_0)$  it follows that $D_{\min}^{\STAB}(\mE(\rho)) \leq D_{\min}^{\STAB}(\rho)$ for any $\mE \in \CSPO$, and for any channel $\mN\in \CPTP(A_0\to A_1)$, it follows that $\Hmin(\Theta[\mN]) \leq \Hmin(\mN)$ for any $\Theta \in \mf_1(A\to B)$ or $\Theta \in \mf_2(A\to B)$. 
The proof for the state case is given below which follows from the data-processing inequality as
\ba
	D_{\min}^{\STAB}(\mE(\rho)) &= \min_{\sigma\in \STAB}D_{\min}(\mE(\rho)\| \sigma)\\
	&\leq \min_{\sigma\in \STAB}D_{\min}(\mE(\rho)\| \mE(\sigma))\\
	&\leq \min_{\sigma\in \STAB}D_{\min}(\rho\| \sigma)\,.
\ea
Proof for the channel case follows similarly.
\item \textit{Sub-additivity.} Sub-additivity holds for both static and dynamic min-relative entropies of magic, i.e., $D_{\min}^{\STAB}(\rho_1 \otimes \rho_2) \leq D_{\min}^{\STAB}(\rho_1) + D_{\min}^{\STAB}(\rho_2)$ for any two density matrices $\rho_1$ and $\rho_2$, and $\Hmin(\mN\otimes \mM) \leq \Hmin(\mN) + \Hmin(\mM)$ for any two quantum channels $\mN$ and $\mM$. Moreover, for single qubit states, the min-relative entropy of magic is additive, i.e.,
$D_{\min}^{\STAB}(\rho_1 \otimes \rho_2) = D_{\min}^{\STAB}(\rho_1) + D_{\min}^{\STAB}(\rho_2)$ \cite{BB+19}. The proof of this is provided in Appendix \ref{append:min_relent_additivity}.
\end{enumerate}
\subsection{Geometric measure of magic for states}
In this subsection, we formally define the geometric measure of magic which to the best of our knowledge has not been defined before.
Inspired from the geometric measure of entanglement \cite{WG03}, we define the geometric measure of magic for pure states as
\begin{align}
    g(\psi) &= 1 - \max_{\phi\in \STAB}\tr[\psi\phi] 
\end{align}
For general mixed states, we can extend the above measure using fidelity as
\begin{equation}
   g(\rho) = 1 - \max_{\sigma \in \STAB}F^2(\rho,\sigma)
\end{equation}
where $F(\rho, \sigma):=\tr\left[\sqrt{\sqrt{\sigma}\rho\sqrt{\sigma}}\right]$ is the fidelity between two states $\rho$ and $\sigma$.
Below we list the properties of this measure:
\begin{enumerate}
    \item \textit{Faithfulness:} $g(\rho)=0$ if and only if $\rho\in \STAB$.
    \item \textit{Monotonicity:} $g(\mE(\rho))\leq g(\rho)\; \forall \; \mE \in \CSPO$. The proof is similar to the proof of monotonicity of geometric measures in \cite{CG19}.
    
    
    \item \textit{Subadditivity:} $g(\rho_1 \otimes \rho_2) \leq g(\rho_1) + g(\rho_2)$. 
    This follows easily if we let $\sigma_1$ and $\sigma_2$ be the respective optimal stabilizer states such that $g(\rho_1) = 1 - F^2(\rho_1,\sigma_1)$ and $g(\rho_2) = 1 - F^2(\rho_2,\sigma_2)$. Then
    \begin{align}
    	\max_{\sigma\in \STAB} F(\rho_1\otimes \rho_2, \sigma)&= \max_{\sigma}\tr\[\sqrt{\sqrt{\sigma}(\rho_1\otimes \rho_2)\sqrt{\sigma}}\]\\
    	&\geq \tr\[\sqrt{\(\sqrt{\sigma_1} \rho_1 \sqrt{\sigma_1} \)\otimes \(\sqrt{\sigma_2} \rho_2 \sqrt{\sigma_2} \) }\]\\
    	&=F(\rho_1,\sigma_1)F(\rho_2,\sigma_2)
	\end{align}
	where the inequality follows by choosing $\sigma = \sigma_1\otimes \sigma_2$.     
\end{enumerate}

\section{Interconversions}\label{sec:interconversions}
Resource interconversion is one of the central themes of resource theory.
In this section, we discuss the conditions for qubit interconversions under CSPOs in \ref{subsec:qubit_interconv}, 
and the conversion of magic states to channels and vice-versa under CSPO preserving and completely CSPO preserving superchannels in \ref{subsec:cost_and_dist}. 
We also formulated the interconversion distance which is given in Appendix~\ref{subsec:interconv_dist}.

\subsection{Qubit interconversion under CSPOs}\label{subsec:qubit_interconv}

For the resource theory of magic, any pure magic state can be used as a resource to achieve universal quantum computation \cite{BK05}.
The procedure involves distilling a pure magic state from a given magic state and then using few copies of this pure magic state to perform any quantum computation. 
Experimentally, it of interest to distil single qubit magic states, and the common choices are that of the $|T\rangle$ state or the $|H\rangle$ state where:
\eqs{
|T\rangle\langle T| &= \frac{1}{2}\left( I + \left( X + Y\right)/\sqrt{2} \right)\; ,\\
|H\rangle\langle H| &= \frac{1}{2}\left( I + \left( X + Y + Z\right)/\sqrt{3} \right)\; .
}

Here, we are interested in a more general problem of finding whether a given single qubit magic state can be converted to another by repeated application of CSPOs.
Equivalently, we want to find out which set of states on the Bloch sphere can be reached by restricting ourselves to the application of CSPOs on a single qubit magic state.
For multiqubit systems, this problem is an NP-hard problem because the number of stabilizer states increases super-exponentially as we increase the dimension. 
For the qubit case, we use geometry to our advantage and provide the following theorem for the conversion of a state $\rho$ into a state $\sigma$. 
We show that this interconversion problem can be cast as a linear programming feasibility problem.
For the purpose of this theorem, let us define $C(\rho) := \{U\rho U^{\dagger}: U\in {\rm Clifford} \}$ as the set of Clifford equivalent states of $\rho$.
We show in the proof of the theorem below that for a single qubit state $\rho$, the set  $C(\rho)$ contains 24 elements unless the state has additional symmetry, in which case the number of elements are less than 24.
For instance, $C(|0\rangle\langle 0 |)$ contains only 6 elements which are all the pure single-qubit stabilizer states.
\begin{theorem}\label{thm:qubit_interconversion}
Let $A$ be a $(3 \times 31)$ matrix with first $24$ columns being the Bloch vectors of the elements of $C(\rho)$, the next $6$ columns being the Bloch vectors of the pure qubit stabilizer states, and the last column being $(1,1,1)^{\rm T}$.
Let $\mathbf{b}$ be the $(3\times 1)$ Bloch vector corresponding to the state $\sigma$.
Then, the state $\rho$ can be converted to the state $\sigma$ using CSPOs if there exists an $\mathbf{x}\in \mathbb{R}^{31}_+$ such that $A\mathbf{x}=\mathbf{b}$.
\end{theorem}
\begin{remark}
The problem of finding $\mathbf{x}$ such that $A\mathbf{x} = \mathbf{b}$ and $\mathbf{x}\geq 0$ is known as an SDP feasibility problem and can be solved using standard techniques in convex analysis~\cite{cvx, GB08}.
It also has a dual given by the Farkas lemma.
Using the dual of the above feasibility problem, we can say that the state $\rho$ cannot be converted to $\sigma$ if there exists a $\mathbf{y}\in \mathbb{R}^3$ such that $A^{\rm T}\mathbf{y}\geq 0 $ and $\mathbf{b}\cdot\mathbf{y}< 0$.
\end{remark}

\begin{proof}

From \cite{SC19} and Eq.~\eqref{CSPO_Choi_defn}, we know that the normalized Choi matrix of any CSPO is a stabilizer state. 
Let $\mE_{A_0\to A_1}\in \CSPO(A)$ such that both $A_0$ and $A_1$ are single qubit systems.
If we denote a pure two qubit maximally entangled stabilizer state as $\psi^{ent}$ and a single qubit stabilizer state as $\phi$, we can 
write the action of $\mE_A$ on any input $\rho\in\md(A_0)$ as
\begin{widetext}
\eqs{
\mE(\rho_{A_0}) &= \tr_{A_0}\[J^{\mE}_{A}\(\rho_{A_0} \otimes I_{A_1}\)\]\\
			&= |A_0|\( \sum_i p_i \tr_{A_0}\[\psi_i^{ent}\(\rho_{A_0} \otimes 
					I_{A_1}\)\] + \sum_{j,k}p_{j,k}\tr_{A_0}\[\(\(\phi_j\)_{A_0}\otimes 
					\(\phi_k\)_{A_1}\) \(\rho_{A_0} \otimes I_{A_1} \) \]\)\\
			&= \sum_i p_i \mU_i(\rho_{A_0}) + |A_0| \sum_{j,k}p_{j,k}\tr\[\(\phi_j\)_{A_0} \rho_{A_0}\] \(\phi_k\)_{A_1}\\
			&=  \sum_i p_i \mU_i(\rho_{A_0}) + \sum_k q_k \phi_k\, ,\label{eq:convex_comb}
}
\end{widetext}
where $q_k = |A_0|\sum_j p_{j,k}\tr\left[ \phi_j \rho \right]$.
In the above, the second equality follows because any two-qubit stabilizer state can be expressed as a convex combination of pure two-qubit entangled and pure two-qubit separable stabilizer states.
From the above equations, we see that the action of a (qubit input and output) CSPO on a qubit can be represented as a convex combination of the action of completely stabilizer preserving unitary operations and stabilizer replacement channels. (An alternative proof can also be found in~\cite{HHG20}).
Note that for two-qubit states, there are a total of 60 pure stabilizer states of which only 24 are entangled \cite{G06}.
Hence there are only 24 single-qubit unitary gates that are completely stabilizer preserving. 
These unitary gates are listed in Appendix \ref{append:unitary_cspo} and are Clifford unitaries. 
Therefore, any state can be transformed to at the most 24 states (including itself) on the Bloch sphere using these unitary gates.
For a single qubit state, which can be expressed as a vector $(r_1, r_2, r_3)^T$ in the Bloch sphere, its transformations using these unitary gates are given in Appendix \ref{append:unitary_cspo}.
Furthermore, if we view the Bloch sphere as been divided into 8 octants according to $(\pm X,\pm Y, \pm Z)$ and each octant to be further subdivided into three subsets such that for one subset it holds that $|\langle X \rangle|\leq |\langle Y\rangle|,|\langle Z\rangle|$, for second subset it holds that  $|\langle Y \rangle|\leq |\langle X\rangle|,|\langle Z\rangle|$, and for the third subset we have  $|\langle Z \rangle|\leq |\langle X\rangle|,|\langle Y\rangle|$, then using table~\ref{Bloch_vector_transformations} in Appendix~\ref{append:unitary_cspo}, we can say that any arbitrary state in some subset (of an octant) is Clifford equivalent to a state in any other subset.
Therefore, we can conclude from the equations and the arguments above that the set of states that can be generated from a given state under the action of CSPOs must belong to a convex polytope in the Bloch sphere, the extreme points of which are the Clifford-equivalent states of the given state and the stabilizer states.
Further, if we let $\{\mathbf{r}_i\}$ denote the set of Bloch vectors of the Clifford equivalent state of $\rho$, $\{\mathbf{s}_k\}$ denote the Bloch vectors of the pure single qubit stabilizer states, and $\mathbf{b}$ as the Bloch vector of $\mE(\rho) $, then from Eq~\eqref{eq:convex_comb}, we can write the Bloch vector $\textbf{b}$ as
\eqs{
\textbf{b} = \sum_i p_i \textbf{r}_i + \sum_k q_k \textbf{s}_k\, . 
}
We can now express the above in the form of the equation $A\mathbf{x} = \mathbf{b}$, where the matrix $A$ is a $(3 \times 30)$ matrix consisting of $\mathbf{r}_i$'s and $\mathbf{s}_k$'s as column vectors, and $\mathbf{x}$ is the $(30\times 1)$ vector consisting of non-negative numbers summing to one.
We can include this last condition on $\mathbf{x}$ by inserting a $(1,1,\dots,1)$ row in $A$ thus making $A$, a $(4\times 30)$ matrix.
Therefore, we can now say that a state $\rho$ can be converted to a state $\sigma$ with Bloch vector $\mathbf{b}$ if there exists a vector $\mathbf{x}\in \mathbb{R}^{30}$ such that
\eqs{
A\mathbf{x} &= \mathbf{b},\;{\rm and}\\
\mathbf{x}&\geq 0\, .
}

\begin{remark}
The above interconversion conditions can be expressed and visualized on a Bloch sphere which has been discussed in Appendix~\ref{append:geometrical_interpretation}.
\end{remark}
\end{proof}

\subsection{Cost and Distillation bounds under CSPO preserving and completely CSPO preserving superchannels}\label{subsec:cost_and_dist}

In this subsection, we find bounds on the cost of converting a magic state to a multi-qubit magic channel and the bounds on distilling magic from a quantum channel using both CSPO preserving and completely CSPO preserving superchannels.
For the case of distillation, we focus on distilling pure single qubit magic states because a pure magic state is enough for achieving universality in the magic state model of quantum computation. 
Besides, due to the complexity involved in verifying whether a state is a stabilizer state, we leave the problem of finding the upper bound of cost and lower bound of distillation using completely CSPO preserving superchannels as open.

Since any pure magic state can be used as a resource to perform universal quantum computation, we define the dynamical magic cost of converting a pure magic state to a channel $\mN\in \CPTP(B_0\to B_1)$ under CSPO preserving superchannels  or completely CSPO preserving superchannels as
\begin{widetext}
\begin{equation}
\cost_{\mf_{1(2)}}(\mN_B) = \min\log\{ |A_1|: \Theta[\psi_{A_1}] = \mN_B,\, \psi\in \md(A_1),\, \Theta\in \mf_{1(2)}(A_1\to B)\}\, .
\end{equation}
\end{widetext}
If we want the cost of simulating a channel in terms of a particular magic state $\psi\in \md(A_1)$, we define cost as
\eqs{
\cost_{\psi,\, \mf_{1(2)}}(\mN_B) = \min\big\{n:\Theta[\psi^n]& = \mN_B,\nonumber\\ 
&\Theta\in \mf_{1(2)}(A_1\to B)\big\}\, .
}
Distillation of a pure single qubit magic state $\psi$ from a channel $\mN\in \CPTP(A_0\to A_1)$ using CSPO preserving or completely CSPO preserving superchannels is defined as
\begin{widetext}
\begin{equation}
    \distill^{\epsilon}_{\psi,\mf_{1(2)}}(\mN_A) = \max\{ n: F\left(\Theta[\mN], \psi^n\right) \geq 1 - \epsilon, \, \Theta\in \mf_{1(2)}(A\to B_1)\}\, .
\end{equation}
\end{widetext}
\begin{proposition}
$\cost_{\mf_1}(\mN) \leq \log(|A_1|)$ if for some system $A_1$, we have
\eqs{\max_{\psi\in \md(A_1)}D_{\min}^{\STAB}(\psi_{A_1})\geq \LR(\mN_B)\, }
where $\LR(\mN_B)$ is the log of the robustness of $\mN_B$.
If $\psi$ is a given single qubit magic state, then  it follows that
\begin{equation}
\cost_{\psi,\mf_1}(\mN)\leq\left\lceil \frac{\LR(\mN)}{D_{\min}^{\STAB}(\psi)} \right\rceil\, . \label{state_dependent_cost}
\end{equation}
\end{proposition}
\begin{proof}
Let for some $\psi\in \md(A_1)$, the following is satisfied
\eqs{D_{\min}^{\STAB}(\psi_{A_1})\geq \LR(\mN_B)\, .\label{eq:cost_requirement}}
Now consider the following superchannel $\Theta\in \ms(A_1 \to B)$ whose action on any input state $\rho\in \md(A_1)$ is given as
\begin{equation}
    \Theta[\rho] \coloneqq \tr[\psi\rho]\mN + (1 - \tr[\psi\rho])\mM\, ,
\end{equation}
 where $\mM$ is the optimal CSPO chosen from the definition of the channel robustness, ${\rm R}(\mN)$.
It is easy to verify that $\Theta[\psi] = \mN$.
From Eq.~\eqref{eq:cost_requirement}, we also get that
\eqs{
&-\log\tr[\psi\sigma]\geq \log(1+{\rm R}(\mN)) \quad \forall \sigma \in \STAB(A_1)\, .
}
Hence, for any $\sigma\in \STAB(A_1)$, it holds that $\tr[\psi\sigma]\leq \frac{1}{1+ {\rm R}(\mN)}$, implying that $\Theta\in \mf_1(A_1\to B)$.
Thus, the cost of converting a pure magic state to a magic channel $\mN_B$ using CSPO preserving superchannels is no greater than $\log(|A_1|)$ if $ \max_{\psi\in \md(A_1)}D_{\min}^{\STAB}(\psi_{A_1}) \geq \LR(\mN)\, .$
  
Further, if $\psi$ is a given single qubit pure magic state, then using the additivity of min-relative entropy of magic for qubits, we get 
\begin{equation}
    \cost_{\psi,\mf_1}(\mN_B) \leq \left\lceil \frac{\LR(\mN)}{D_{\min}^{\STAB}(\psi)} \right\rceil\, .
\end{equation}

\end{proof}
\begin{remark}
We numerically verify that the bound in Eq.~\eqref{state_dependent_cost} is not trivial.
As an example, we use the $|T\rangle$ state $(D_{\min}^{\STAB}(|T\rangle\langle T|) = 0.2284)$ to calculate the upper bound of cost of creating some magic states.
We present the comparison of the upper bound of our results of cost with the lower bound obtained in \cite{HC17} as a table below. Note that in \cite{HC17} the free operations were stabilizer operations.
 In the table, a general resource state $|U\rangle = U|+\rangle$ where $|+\rangle$ is the maximally coherent state and $U$ is some unitary gate.
Also, some special states include the $|H\rangle$ state which is the single-qubit state with Bloch vector $(1,1,1)/\sqrt{3}$ and has robustness $\sqrt{3}$, $|\chi\rangle$ state is the two-qubit state with maximum robustness of $\sqrt{5}$ for two-qubit states, and $|{\rm Hoggar}\rangle$ state is the three-qubit state which maximizes robustness for three-qubit states and has robustness $3.8$.

\begin{table}[h]
\centering
\captionsetup{type=table}
\begin{tabular}{|>{\centering\arraybackslash}m{0.25\linewidth}|>{\centering\arraybackslash}m{0.25\linewidth}|>{\centering\arraybackslash}m{0.25\linewidth}|}
 \hline
  State & upper bound from  our work	  & lower bound from \cite{HC17}  \\
 \hline
 
 $ |H\rangle $ & 2 & 2 \\
 $ |CS_{1,2}\rangle $ & 3 & 3\\
 $ |T_{1,2,3}\rangle $ & 4 & 3\\
 $ |\chi\rangle $ & 4 & 4\\
 $ |CCZ\rangle $ & 4 & 4\\
 $ |CS_{12,13}\rangle $ & 4& 4\\
 $ |T_1 CS_{2,3}\rangle $ & 5& 4 \\
 $ |T_1CS_{12,13}\rangle $ & 5& 5\\
 $|{\rm Hoggar}\rangle$ & 6 & 6 \\
  \hline
\end{tabular}
\captionof{table}{Comparison of upper bound of magic cost of states}
\end{table}
\end{remark}
\begin{remark}
We would like to emphasize here that we provide a general result for the case of channels by giving a precise formula to find the upper bound on the cost that depends on the log-robustness of magic of the channel and the min-relative entropy of magic of the single-qubit state.
\end{remark}
\begin{proposition}
The cost of converting a pure magic state $\psi_{A_1}\in\md(A_1)$ to a target channel $\mN_B\in \CPTP(B_0\to B_1)$ using CSPO preserving or completely CSPO preserving superchannels is lower bounded by
\begin{equation}
    \dfrac{\LR_g(\mN_B)}{\LR_g(\psi_{A_1})}\leq\cost_{\psi,\mf_{1(2)}}(\mN_B)\,.
\end{equation}
\end{proposition}
\begin{proof}
The proof follows from the standard resource theoretic methods and can be seen as a special case of theorem 1 of \cite{RT21} together with the sub-additivity of generalized robustness of magic.
\end{proof}

\begin{proposition}\label{upper_bound_distillation}
Given a channel $\mN\in \CPTP(A_0\to A_1)$ and a single qubit state $\psi$, the following holds
\begin{equation}
    {\rm DISTILL}_{\psi,\mf_{1(2)}}(\mN_A)\leq \dfrac{\Hmin(\mN_A)}{D_{\min}^{\STAB}(\psi)}\, .
\end{equation}
\end{proposition}
\begin{proof}
The proof of the above proposition also follows from standard resource theoretic methods \cite{RT21,YZ+20} and the additivity of min-relative entropy of magic. For completeness, we provide the proof in Appendix \ref{upper_bound_distillation_proof}.
\end{proof}

\begin{proposition}
The lower bound on distilling a single qubit pure magic state $\psi$ from a channel $\mN\in \CPTP(A_0\to A_1)$ using a CSPO preserving superchannel is given by
\begin{equation}
    \distill^{\epsilon}_{\psi,\mf_1}(\mN_A) \geq\left\lfloor \frac{D_{\min}^{\epsilon,\,\STAB}(\Tilde{J}^{\mN}_A)}{\LR(\psi)} \right\rfloor\, ,
\end{equation}
where $\Tilde{J}^{\mN}_A$ is the normalized Choi matrix of the channel $\mN$, and $D_{\min}^{\epsilon,\,\STAB}(\cdot)$ represents the hypothesis testing relative entropy of magic which we have defined in Appendix~\ref{append:hyptesting_rel_ent_magic}.
\end{proposition}
\begin{proof}
Let $n$ be the largest non-negative integer such that $D_{\min}^{\epsilon,\,\STAB}(\tJ^{\mN}_A) \geq n\LR(\psi)$.
Then, we can construct the following superchannel $\Theta\in \ms(A\to B_1)$ such that for any input channel $\mM\in \CPTP(A_0\to A_1)$ 
    \begin{equation}
        \Theta[\mM] \coloneqq \tr[\Tilde{J}^{\mM}_A E] \psi^{n} + (1 - \tr[\Tilde{J}^{\mM}_A E])\sigma\, ,
    \end{equation}
    where $\sigma\in \STAB(B_1)$ is chosen from the definition of $R(\psi^n)$, and $E$ is the optimal POVM element chosen in the definition of hypothesis testing relative entropy of magic, $D_{\min}^{\epsilon,\,\STAB}(\Tilde{J}^{\mN})$.
    We first notice that for such a superchannel
\begin{align}
    F(\Theta[\mN],\psi^n) &\geq \tr[\Theta[\mN]\psi^n]\\
    &\geq \tr[\Tilde{J}^{\mN} E]\\
    &\geq 1 - \epsilon
\end{align}
where the last inequality comes from the fact that $E$ is optimal in $D_{\min}^{\epsilon,\,\STAB}(\Tilde{J}^{\mN})$.

Since $D_{\min}^{\epsilon,\,\STAB}(\tJ^{\mN}_A) \geq n\LR(\psi)$, we get
\eqs{
-\log\tr[E\sigma] &\geq \log(1+ {\rm R}(\psi))^n
\geq \log(1 + {\rm R}(\psi^{n}))
}
for all $\sigma\in \STAB(A_0A_1)$.
Therefore, if the input $\mM\in \CPTP(A_0\to A_1)$ is a CSPO, then 
$-\log\tr[E\tJ^{\mM}_A] \geq \log(1 + {\rm R}(\psi^{n}))$
which implies that
$$\tr[E\tJ^{\mM}_A]\leq \frac{1}{1+ {\rm R}(\psi^n)}\,.$$
Hence, $\Theta$ is a CSPO preserving superchannel.
Thus, we can distill atleast n copies of the single qubit state $\psi$ from the channel $\mN$ where $n$
satisfies $D_{\min}^{\epsilon,\,\STAB}(\tJ^{\mN}_A) \geq n\LR(\psi)$. 
\end{proof}

\section{Classical Simulation Algorithm for circuits}\label{sec:classical_algo}
The goal of a classical simulation algorithm is to estimate Born rule probabilities or to find the expectation value of an observable. 
To this purpose, a class of algorithms, known as the quasiprobability simulation techniques, have recently been developed that make use of the quasiprobability decomposition of magic states or channels \cite{SRP+20,HC17,SC19,WWS20}.
The runtime of these algorithms has been shown to be of the order of the square of the robustness \cite{HC17, SC19}, or the square of another similar monotone, the dyadic negativity~\cite{SRP+20}.
In \cite{SRP+20}, another simulation technique, the constrained path simulator for states was introduced with the idea to reduce the runtime of the simulation. 
This simulation technique offers constant runtime by compromising with the precision in estimating the expected value.

Below, we extend the constrained path simulator algorithm to the general case of a circuit composed of a sequence of channels acting on an initial stabilizer state and ending with a measurement of some Pauli observable.
We modify the algorithm so that we achieve the estimate with a precision more than or equal to a desired precision.
With this modification, the runtime of the algorithm is not a constant but depends on the desired precision (or the desired error). For any non-zero error, the runtime  never exceeds that of a quasiprobability simulator for channels. Moreover, if there is no bound on the error/precision, the algorithm achieves a constant runtime.


The overall idea of the constrained path simulator for states is as follows. 
A magic state $\rho\in \md(A_1)$ can be decomposed as $\rho = t\sigma_+ - (t-1)\rho_-$ for some $t\geq 1$, $\sigma_{+}\in \STAB(A_1)$, and $\rho_-\in \md(A_1)$.
The constrained path simulator for states works by constraining the quasiprobability decomposition of a state to the positive part, i.e., by making the approximation $\rho\approx t\sigma_+$. 
Then, the algorithm estimates $t\tr[E\mathcal{O}(\sigma_+)]$ upto $\epsilon$ error using a Clifford simulator (like quasiprobability simulator).
Here, $E$ is some Pauli observable, and $\mathcal{O}$ is a CSPO.
This estimate is then used to obtain the expectation value $\tr[E\mathcal{O}(\rho)]$ and the estimation error.
The runtime of the algorithm is decided by the Clifford simulator used.
By defining $\epsilon$ as the product of a constant $c$ and $t$, the algorithm was shown to have a constant runtime.

\paragraph*{Constrained path simulator for channels.}

Let $\mN$ be a circuit composed of a sequence of $n$ channels and let the $i^{\rm th}$ circuit element be denoted by $\mN_i$.
As mentioned previously in remark~\ref{remark:channel_decomp}, the circuit element $\mN_i$ can be decomposed using some CSPO $\mE_i$ and some other channel $\mM_i$ such that $\mN_i = \lambda_i\mE_i - (\lambda_i - 1)\mM_i$ where $\lambda_i$ is the generalized robustness of $\mN_i$.
Then, for the whole circuit we can write
\begin{widetext}
\ba
\mN=\mN_n\circ \cdots \circ \mN_1  &= (\lambda_n\cdots\lambda_1)(\mE_n\circ\cdots\circ\mE_1 ) + \ldots + ((\lambda_n -1)\cdots(\lambda_1 - 1)) \mM_n\circ\cdots\circ\mM_1 \\
&=(\lambda_n\cdots\lambda_1)(\mE_n\circ\cdots\circ\mE_1 ) +  ((\lambda_n\cdots\lambda_1)-1)\mM\\
&= \lambda \mE + (\lambda - 1)\mM
\ea
\end{widetext}
where $\lambda = \lambda_n\cdots\lambda_1$, $\mE = \mE_n\circ\cdots\circ\mE_1$ and $\mM$ follows from simple arithmetic manipulation of the first equation and is the probabilistic combination of the sequence of channels where each sequence contains atleast one $\mM_i$.
The aim of the algorithm is to estimate $\tr[E\mN(|0\rangle\langle 0|)]$ with a precision more than or equal to some target precision and a runtime less than what can be achieved by a quasiprobability simulator.

The algorithm starts by replacing the original circuit $\mN$ with another circuit $\mN'$ to achieve the mean estimate up to some target error $\Delta^*$.
The algorithm first replaces the channel $\mN_j$ with $\lambda_j\mE_j$ if $\lambda_j$, the generalized robustness of $
\mN_j$, is less than some fixed real number $\lambda^*$.
Here, $\mE_j$ is the optimal CSPO such that $\lambda_j\mE_j \geq \mN_j$.
The choice of $\lambda^*$ ensures that the estimation error never exceeds the target allowed error.
Then, using the static Monte Carlo routine introduced in \cite{SC19} for circuits, the algorithm estimates $\lambda'\tr[E\mN'(|0\rangle\langle 0|)]$ up to $\epsilon$ error where $\lambda'$ is the product of the generalized robustnesses of the replaced channels and the error $\epsilon$ equals a constant $c$ multiplied with $\lambda'$.
Next, using  $\epsilon$, $\lambda'$, and the estimate we obtained above, the algorithm outputs the estimate of the expectation value $\tr[E\mN(|0\rangle\langle 0|)]$ up to error $\Delta\leq \Delta^*$ following some trivial steps.

In the static Monte Carlo routine, the runtime of the algorithm is decided by finding the total number $N$ of steps required to achieve the mean estimate up to an additive error $\epsilon$ with success probability $1-p_{\rm fail}$. 
The number of steps $N$ that the static Monte Carlo takes is given by
\begin{align}
    N = \left\lceil 2 \epsilon^{-2}\|q\|_1^2\log(2p_{\rm fail}^{-1}) \right\rceil
\end{align}
where $\|q\|_1 = \prod_j{\rm R}(\mN_j)$ and ${\rm R}(\mN_j)$ is the robustness of the circuit element $\mN_j$ as defined in \cite{SC19}.
In our hybrid algorithm, since we choose to keep some channels and replace some with CSPOs, the number of steps to estimate  $\lambda\tr[E\mN'(|0\rangle\langle 0|)]$ upto $\epsilon$ error with success probability $1 - p_{\rm fail}$ is given by
\begin{align}
    N &= \left\lceil 2\epsilon^{-2}\lambda^2\prod_{j:\lambda_j>\lambda^*}{\rm R}(\mN_j)^2\log(2p_{\rm fail}^{-1})\right\rceil\\
    &= \left\lceil 2c^{-2}\prod_{j:\lambda_j>\lambda^*}{\rm R}(\mN_j)^2\log(2p_{\rm fail}^{-1})\right\rceil
\end{align}
where $c$ is a pre-defined small constant. In this sense, the number of steps only depend on the robustness of the channels whose $\lambda_i > \lambda^*$.
Note that if all the channels are selected by the algorithm, we essentially have the runtime as that of static Monte Carlo routine.
If all the channels are replaced in the initial steps then we get a constant runtime.

\begin{algorithm}
\renewcommand{\thealgorithm}{}
\caption{Dynamic constrained path simulator}
\flushleft{\textbf{Input:}
    (i) Sequence of channels $\mN_1,\ldots, \mN_n$ such that the target channel $\mN = \mN_n\circ \cdots \circ \mN_1$.
    (ii) Real numbers $0< c, p_{\rm fail}<<1$ and Pauli observable $E$.
    (iii) Desired error $\Delta^*$.\\}
 \textbf{Pre-Computation:}
    (i) $\lambda^* = (\Delta^* + 1)^{1/n}$.
    (ii) For each circuit element, an optimal decomposition in terms of CSPOs is determined.\\
 \textbf{Output:}
    (i) Born rule probability estimate $\hat{E}$.
    (ii) Error $\Delta$ such that,
    $|\hat{E} - \tr\left[E\, \mN(|0\rangle\langle 0|)\right]|\leq \Delta$, and $\Delta\leq \Delta^*$. \\
	\begin{algorithmic}[1]
  		\FOR{$i$ $\leftarrow$ $1$ to $n$}
        	    \STATE$\lambda_i\leftarrow \Lambda^+(\mN_i)$, and denote the optimal free channel by $\mE_i$.
            	\IF{$\lambda_i\leq \lambda^*\;$:}
                	\STATE $\mN_i\leftarrow\lambda_i \mE_i$
            	\ENDIF
    	\ENDFOR
    	\STATE $\mN' \leftarrow \left(\prod_{j:\lambda_j\leq \lambda^*}\lambda_j\right)\left(\mF_n\circ\cdots\circ \mF_1\right)$, where $\mN'$ denotes the new circuit that will be used to find the estimate and $\mF_k$'s denote the circuit elements given be 
    	\begin{equation}
        \mF_k=
        \begin{cases}
            \mE_k\quad {\rm if}\; \lambda_k \leq \lambda^*\\
            \mN_k\quad {\rm otherwise}
        \end{cases}
    	\end{equation}
    	\STATE $\epsilon\leftarrow c\lambda$ where $\lambda = \prod_{j:\lambda_j\leq \lambda^*}\lambda_j$
    	\STATE Let $E_{\mN'}$ be an estimate of $\lambda\tr[E\mN'(|0\rangle\langle 0|)]$ upto $\epsilon$ error and success probability $1 - p_{\rm fail}$.
    	\STATE $E_{\max}\leftarrow \min\{ 1, E_{\mN'}+ \epsilon + \lambda - 1\}$
    	\STATE $E_{\min}\leftarrow \max\{ -1, E_{\mN'}- \epsilon - \lambda + 1\}$
    	\STATE $\hat{E}\leftarrow (E_{\max}+E_{\min})/2$
    	\STATE$\Delta\leftarrow (E_{\max}-E_{\min})/2$
	\end{algorithmic}
\end{algorithm}

\textbf{Analysis}

As with the constrained path simulator for states, the choice of $E_{\max}$ and $E_{\min}$ ensure that for all $\lambda$ and $E_{\mN'}$, the following inequality holds with probability $1 - p_{\rm fail}$
\begin{equation}
    |\hat{E} - \tr[E\mN(|0\rangle\langle 0|)]|\leq \Delta
\end{equation}
To justify the choice of $\lambda^*$, let $\lambda^*$ be the generalized robustness of each channel used in the circuit and $\lambda^*$ times the optimal CSPO for each channel is considered in the above routine.
Then, for any $\Delta$ we have
\begin{align}
    \Delta \leq \lambda(1 + c) - 1
\end{align}
and hence we require $\lambda(1+c) - 1 \leq \Delta^*$. Assuming there are $n$ channels in the circuit, we get
\begin{align}
    \lambda^* &\leq \left(\frac{\Delta^* + 1}{1 + c}\right)^{1/n}\\
    &\approx (\Delta^* + 1)^{1/n}
\end{align}
Since this is the worst-case analysis, in practical scenarios we will have $\lambda \leq \lambda^*$(equality only arising when the circuit consists of just one channel applied $n$ times), and therefore
$\Delta\leq \Delta^*$.

\section{Conclusion}

In this work, we developed the dynamical multi-qubit resource theory of magic of quantum channels by identifying the completely stabilizer preserving operations (CSPOs) as the set of free operations. 
CSPOs are a perfect candidate for the free channels of a dynamical resource theory of magic because they form the largest known set of operations that cannot provide any quantum advantage.
In previous resource theoretic studies of magic channels, the superchannel approach was only taken in \cite{WWS19} where the authors considered the odd-dimensional qudit case and the free channels were the completely positive Wigner preserving operations (CPWPO). 
There, the free superchannels were chosen to be the ones that completely preserve the set of CPWPO.
In this paper, we defined and characterized two sets of free superchannels - namely, the CSPO preserving superchannels and the completely CSPO preserving superchannels.
We characterized completely CSPO preserving superchannels in terms of their Choi matrices, and in particular, we showed that a superchannel is completely CSPO preserving if and only if its normalized Choi matrix is a stabilizer state.
We then defined monotones for states and channels which include the generalized robustness of magic for channels, the min-relative entropy of magic for channels, and the geometric measure of magic for states. 
We also addressed some resource interconversion problems, specifically proving that the qubit interconversion under CSPOs can be solved with simple linear programming.
We then determined a closed formula for the upper and lower bound on both the cost of simulating a channel from a qubit and distilling a qubit  magic state from a channel, under CSPO preserving superchannels. We also formulated the lower bound on the qubit cost of simulating a magic channel, and the upper bound on distilling a pure qubit magic state from a magic channel under completely CSPO preserving superchannels using the standard resource theoretic techniques.
Finally, we gave a classical simulation algorithm to find expectation values given a general quantum circuit. 
The algorithm works by selecting and replacing some circuit elements with some CSPO, based on a parameter that depends on the minimum target precision required.
Hence, due to this selective replacement algorithm, the runtime of our algorithm also depends on the precision required.
If the precision required is too tight, then the runtime reaches that of the static Monte Carlo simulation algorithm given in \cite{SC19}, whereas, if there is no bound on the precision, the algorithm has a constant runtime and can be seen as a generalization of the constrained path simulator introduced in \cite{SRP+20} for states.
These classical simulation algorithms help benchmark the quantum computational speedup and there is a lot left to explore in the general circuit case.
Apart from that, it would be interesting to explore non-deterministic transformations and catalytic transformations under CSPO preserving and completely CSPO preserving superchannels.
Lastly, because of the difficulty in verifying whether a state is a stabilizer or not, we were unable to find lower bounds on distilling magic using completely CSPO preserving operations and leave it as an open problem.

\section*{Acknowledgements}
G.S. would like to thank Bartosz Regula, James Seddon, Carlo Maria Scandolo, and Yunlong Xiao for helpful discussions. G.S. would also like to thank the anonymous QIP 2022 reviewers for their useful comments.
G.G. is supported by the Natural Sciences and Engineering Research Council (NSERC) of Canada.
\bibliography{ref}

\onecolumngrid
\begin{appendix}

\section{Superchannels}\label{superchannel_desription}

A superchannel is a linear map that takes a quantum channel to another quantum channel.
In other words, we can say that a superchannel $\Theta$ describes the evolution of a quantum channel $\mN\in \CPTP(A_0\to A_1)$ to a target channel $\mM\in \CPTP(B_0\to B_1)$ as
\begin{equation}\label{append:channel_evol}
 \Theta_{A\to B}[\mN_A] = \mM_B
\end{equation}
and even when acting on part of the channel as
\eqs{\1_{R}\otimes \Theta_{A\to B}[\mN_{AR}] = \mM_{BR}
}
where $\mN_{AR}\in \CPTP(A_0R_0 \to A_1R_1)$, $\mM_{BR}\in \CPTP(B_0R_0 \to B_1R_1)$, and  $\1_R$ denotes the identity superchannel that takes the dynamical system $R$ to $R$.
A superchannel can be realized in terms of a pre- and a post-processing channel.
Let $\mE\in \CPTP(B_0\to E_1A_0)$ be the pre-processing channel and $\mF\in \CPTP(E_1A_1\to B_1)$ be the post-processing channel for a superchannel $\Theta\in \ms(A\to B)$, then the LHS describing the evolution in Eq.~\eqref{append:channel_evol} can be written as
\eqs{
\Theta_{A\to B}[\mN_A] = \mF_{E_1A_1\to B_1}\circ\mN_{A_0\to A_1}\circ\mE_{B_0\to E_1A_0}
}
Apart from that the transformation of Eq.~\eqref{append:channel_evol} can also be written using Choi matrices of channels $\mN$, $\mM$, and the superchannel $\Theta$ as
\eqs{
J^{\mM}_B = \tr_{A}\[\J^{\Theta}_{AB}\(\(J^{\mN}_{A}\)^{T}\otimes I_B \) \]
}
where the Choi matrix of a superchannel is defined in terms of a linear map $\mQ^{\Theta}$
\eqs{
\J^{\Theta}_{AB} = \mQ^{\Theta}_{\tilde{A}_1\tilde{B}_0\to A_0B_1}\(\phi^+_{A_1\tilde{A}_1}\otimes \phi^+_{B_0\tilde{B}_0} \)
}
where the linear map takes bounded operators in $\tilde{A}_1\tilde{B}_0$ to bounded operators in $A_0B_1$. More details about supermaps and superchannels can be found in \cite{CDP08, CDP09, CDP+13, BP19, BGS+20, G19}.

Lastly note that the Choi matrix of a superchannel $\Theta\in \ms(A\to B)$ follows the following conditions \cite{G19}:
\eqs{
\J^{\Theta}_{AB}&\geq 0\; ,\\
\J^{\Theta}_{A_1B_0} &= I_{A_1B_0}\; ,\\
\J^{\Theta}_{AB_0} &= \J^{\Theta}_{A_0B_0}\otimes \dfrac{I_{A_1}}{|A_1|}\; .
}

\section{Interconversion Distance}\label{subsec:interconv_dist}
We define the interconversion distance from a state $\rho\in \md(A_0)$ to another state $\sigma\in \md(A_1)$ as
\begin{align}
    d(\rho_{A_0} \to \sigma_{A_1}) &= \frac{1}{2}\min_{\mE \in \CSPO(A_0\to A_1)} \|\mE(\rho) - \sigma \|_1\\
    &=\min_{\mE \in \CSPO}\left( \max_{0\leq P \leq I} \tr\left[\left(\mE(\rho) - \sigma\right)P\right] \right)
\end{align}
Using the dual of trace norm, we can express the above interconversion distance as follows
\begin{align}
    d(\rho\to \sigma) = \min &\tr[X+Y]\\
    \rm{s.t.} & 
    \begin{pmatrix}
        X & \mE(\rho) - \sigma\\
        \mE(\rho) - \sigma & Y
\end{pmatrix}\geq 0\, ,\\
& X\geq 0\,,\,\, Y\geq 0\,,\\
& J^{\mE}_{A_0A_1}\geq 0\, ,\,\, J^{\mE}_{A_0} = I_{A_0}\, ,\\
&\frac{J^{\mE}_{A_0A_1}}{|A_0|}\in \STAB 
\end{align}

\section{Proof of additivity of min-relative entropy of magic for qubits}\label{append:min_relent_additivity}

To prove the additivity of min-relative entropy of magic for qubits, first note that the projector onto the support of a qubit state is identity if the state is mixed, else it is the state itself if it is pure.
For the proof, we construct the following four possible cases for qubits $\rho_1$ or $\rho_2$
\begin{enumerate}
    	\item For $\rho_1, \rho_2 > 0$, we  get
    	\begin{align}
     	   D_{\min}^{\STAB}(\rho_1\otimes \rho_2) &= -\log_2 \max_{\psi\in \overline{\STAB}} \tr[(P_{\rho_1}\otimes P_{\rho_2})\psi]\\
    	    &= -\log_2 \max\tr[(I\otimes I)\psi]\\
     	   &=0\\
     	   &=D_{\min}^{\STAB}(\rho_1) +  D_{\min}^{\STAB}(\rho_2)
    	\end{align}
    	\item For $\rho_1 > 0$ and $\rho_2 = |\chi\rangle\langle\chi |$, we get
    	\eqs{
    	    D_{\min}^{\STAB}(\rho_1\otimes \rho_2) &= -\log_2 \max_{\psi\in \overline{\STAB}(A_1 A_2)} \tr[(P_{\rho_1}\otimes P_{\rho_2})\psi]\\
      	  &= -\log_2 \max_{\phi\in \overline{\STAB}(A_2)}\tr[|\chi\rangle\langle\chi |\phi]\\
      	  &= D_{\min}^{\STAB}(\rho_2)\\
      	  &=D_{\min}^{\STAB}(\rho_1)+D_{\min}^{\STAB}(\rho_2)
    	}
    	\item For $\rho_1 = |\chi\rangle\langle\chi|$ and $\rho_2>0$, we get the same result as obtained in 2, i.e., 
  	  \begin{align}
  	      D_{\min}^{\STAB}(\rho_1\otimes \rho_2) = D_{\min}^{\STAB}(\rho_1)+D_{\min}^{\STAB}(\rho_2)
 	   \end{align}
  	  \item For the case when both $\rho_1$ and $\rho_2$ are pure and let $\rho_1 = |\chi\rangle\langle\chi |$ and $\rho_2 = |\omega\rangle\langle\omega|$, we get
   	 \begin{align}
   	     D_{\min}^{\STAB}(\rho_1\otimes \rho_2) &= -\log_2\max_{\psi\in \overline{\STAB}}\tr[(|\chi\rangle\langle \chi|\otimes |\omega\rangle\langle\omega|)\psi]\\
    	    &= -\log_2 F(|\chi\rangle\langle \chi|\otimes |\omega\rangle\langle\omega|)\\
        	&=-\log_2 \left( F(|\chi\rangle\langle \chi|)F(|\omega\rangle\langle\omega|) \right)\\
        	&= D_{\min}^{\STAB}(\rho_1)+D_{\min}^{\STAB}(\rho_2)
    	\end{align}
    	where the second equality follows from the definition of stabilizer fidelity as defined in \cite{BB+19}. The third equality follows from Theorem 5 and Corollary 3 of \cite{BB+19}. 
	\end{enumerate}
	Therefore, for single-qubit states we find that the min-relative entropy of magic is additive.

\section{Robustness of magic}\label{append:rob_of_magic}
We define the robustness of magic of a quantum state as
\begin{equation}
    {\rm R}(\rho) = \min\left\{ \lambda \geq 0 : \frac{\rho + \lambda \sigma}{\lambda + 1}\in \STAB,\; \sigma \in \STAB\right\}
\end{equation}
which is slightly different from how it was originally defined in \cite{HC17}. We use this definition because any resource monotone must be zero for free elements.
Likewise, we define channel robustness of magic of a quantum channel $\mN$ as
\begin{equation}
    {\rm R}(\mN_A) = \min\left\{ \lambda \geq 0: \frac{\mN + \lambda \mE}{ \lambda + 1}\in \CSPO,\; \mE\in \CSPO\right\}
\end{equation} 
which agains differs slightly from the definition of channel robustness of magic in \cite{SC19}.

Both these quantities are magic monotones and are sub-multiplicative under tensor products.
Therefore, the log of the robustness of magic (denoted as $\LR$) is sub-additive i.e.,
\begin{align}
    \LR(\rho^{\otimes m}) &\leq m\LR(\rho)\, ,\\
    \LR(\mN^{\otimes m}) &\leq m\LR(\mN)   \, . 
\end{align}
where $\LR(\rho) = \log(1+ {\rm R}(\rho))$ and $\LR(\mN) = \log(1+ {\rm R}(\mN))$.

\section{Hypothesis testing relative entropy of magic}\label{append:hyptesting_rel_ent_magic}

The hypothesis testing relative entropy of magic or the operator smoothed min-relative entropy of magic is defined as
\begin{align}
    D_{\min}^{\epsilon,\,\STAB}(\rho) &= \min_{\sigma\in \STAB}D^{\epsilon}_{\min}(\rho\|\sigma)\\
    &= \min_{\sigma\in \STAB}(-\log \min \tr[E\sigma]\\
    &\qquad\qquad\qquad\quad\; {\rm s.t.}\;\; 0\leq E\leq I,\\
    &\qquad\qquad\qquad \qquad\quad\, \tr[E\rho]\geq 1 - \epsilon\;)
\end{align}
For $\epsilon=0$, the hypothesis testing relative entropy of magic becomes equal to the min-relative entropy of magic, i.e., $D_{\min}^{\epsilon = 0,\,\STAB}(\rho)  = D_{\min}^{\STAB}(\rho)$.

\section{Proof of proposition \ref{upper_bound_distillation}} \label{upper_bound_distillation_proof}
First we note that for any $\mE_A\in \CPTP(A)$ and any $\Theta\in \CSPSC(A\to B_1)$ we have
    \begin{align}
        D_{\min}\(\psi^{\otimes k}\| \Theta[\mE]\) &= -\log_2 \tr\[\psi^{\otimes k} \Theta[\mE]\]\\
        &\geq D_{\min}^{\STAB}(\psi^{\otimes k})\\
        &= k\, D_{\min}^{\STAB}(\psi)
    \end{align}
    where the inequality follows from the definition of min-relative entropy of magic for states and the last equality follows from its additivity for single-qubit states.

 The hypothesis testing relative entropy \cite{BD10,WR12} between two states $\rho_1$ and $\rho_2$ is given by
    \begin{equation}
        D^{\epsilon}_{\rm Hyp}(\rho_1\| \rho_2) \coloneqq -\log_2 \min \{\tr[M\rho_2]: 0\leq M\leq I\,,\; \tr[M\rho_1]\geq 1-\epsilon \}\, .
    \end{equation}
     and its channel counterpart can be given as
     \begin{equation}
     D^{\epsilon}_{\rm Hyp}\(\mN_A\|\mM_A\)\coloneqq \sup_{\psi\in \md(R_0 A_0)}D^{\epsilon}_{\rm Hyp}\( \mN(\psi_{R_0A_0})\|\mM(\psi_{R_0 A_0}) \)
     \end{equation}
Using this definition, we then have
\eqs{
k\, D_{\min}^{\STAB}(\psi)&\leq \min_{\mE\in CSPO} D_{\min}(\psi^{\otimes k}\| \Theta[\mE])\\
        &\leq \min_{\mE\in CSPO} D_{\rm Hyp}^{\epsilon}(\Theta[\mN]\| \Theta[\mE])\\
        &\leq \min_{\mE\in CSPO} D_{\rm Hyp}^{\epsilon}(\mN\| \mE)
}
where the second inequality follows from the definition of hypothesis testing relative entropy and the last inequality follows from the data-processing inequality.
And therefore, we get
\begin{equation}
{\rm DISTILL}^{\epsilon}_{\psi}(\mN_A)\leq \dfrac{\min_{\mE\in \CSPO(A_0\to A_1)}D^{\epsilon}_{\rm Hyp}(\mN\|\mE)}{D_{\min}^{\STAB}(\psi)}
\end{equation}
which for exact distillation process (i.e., $\epsilon = 0$) will become
\begin{equation}
{\rm DISTILL}_{\psi}(\mN_A)\leq \dfrac{\Hmin(\mN_A)}{D_{\min}^{\STAB}(\psi)}
\end{equation}
\section{Single qubit Unitary CSPOs}\label{append:unitary_cspo}
Table~\ref{unitaryCSPO_and_Choi} lists the set of 24 unitary gates which are completely stabilizer preserving along with corresponding (unnormalized) Choi matrices.
Table~\ref{Bloch_vector_transformations} gives an account of the states generated by these unitary CSPOs. Since a single qubit state can be represented as a vector $(r_1,r_2,r_3)^T$ in the Bloch sphere, we will give below the vectors to which this vector transforms on the application of the above unitaries.

\noindent\begin{minipage}{.45\linewidth}
  \centering
\captionsetup{type=table}
  \begin{tabular}{ |c|c| } 
 \hline
 Unitary gate & state corresponding to\\
  &associated Choi matrix  \\
 \hline
 $I$ & $|00\rangle + |11\rangle$  \\ 
 $X$ & $|01\rangle + |10\rangle$  \\
 $Z$ & $|00\rangle - |11\rangle$  \\
 $XZ$ & $|01\rangle - |10\rangle$  \\
 \hline
 $H$ & $|0+\rangle + |1-\rangle$  \\
 $HX$ & $|0-\rangle + |1+\rangle$  \\
 $HZ$ & $|0+\rangle - |1-\rangle$  \\
 $HXZ$ & $|0-\rangle - |1+\rangle$  \\
 \hline
 $S$ & $|00\rangle + i|11\rangle$  \\
 $XS$ & $|01\rangle + i|10\rangle$  \\
 $ZS$ & $|00\rangle - i|11\rangle$  \\
 $XZS$ & $|01\rangle - i|10\rangle$  \\
 \hline
 $HS$ & $|0+\rangle + i|1-\rangle$  \\
 $HSZ$ & $|0+\rangle -i|1-\rangle$  \\
 $HXS$ & $|0-\rangle +i|1+\rangle$  \\
 $HXSZ$ & $|0-\rangle -i|1+\rangle$  \\
 \hline
 $SH$ & $|0\,+i\rangle + |1\,-i\rangle$  \\
 $SHZ$ & $|0\,+i\rangle - |1\,-i\rangle$  \\
 $SHX$ & $|0\,-i\rangle + |1\, +i\rangle$  \\
 $SHXZ$ & $|0\, -i\rangle - |1\, +i\rangle$  \\
 \hline
 $SHS$ & $|0\,+i\rangle + i|1\,-i\rangle$  \\
 $SHSZ$ & $|0\,+i\rangle -i |1\,-i\rangle$  \\
 $SHSX$ & $i|0\,-i\rangle + |1\,+i\rangle$  \\
 $SHSXZ$ & $i|0\,-i\rangle - |1\,+i\rangle$  \\
 \hline
\end{tabular}
\captionof{table}{Unitary CSPOs and their Choi matrices.}
\label{unitaryCSPO_and_Choi}
\end{minipage}\hfill
\begin{minipage}{.45\linewidth}
  \centering
  \captionsetup{type=table}
  \begin{tabular}{ |c|c| } 
 \hline
 Unitary gate & Transformed vector   \\
 \hline
 $I$ & $r_1,r_2,r_3$\\
 $SH$ & $r_2,r_3,r_1$\\
 $HSZ$ & $r_3,r_1,r_2$\\
 \hline
 $X$ & $r_1,-r_2,-r_3$\\
 $SHXZ$ & $r_2,-r_3,-r_1$\\
 $HS$ & $r_3,-r_1,-r_2$\\
 \hline
 $Z$ & $-r_1,-r_2,r_3$\\
 $SHX$ & $-r_2,-r_3,r_1$\\
 $HXSZ$ & $-r_3,-r_1,r_2$\\
 \hline
 $Y$ & $-r_1,r_2,-r_3$\\
 $SHZ$ & $-r_2,r_3,-r_1$\\
 $HXS$ & $-r_3,r_1,-r_2$\\
 \hline
 $SHS$ & $r_1,r_3,-r_2$\\
 $HZ$ & $r_3,r_2,-r_1$\\
 $XZS$ & $r_2,r_1,-r_3$\\
 \hline
 $SHSX$ & $r_1,-r_3,r_2$\\
 $H$ & $r_3,-r_2,r_1$\\
 $ZS$ & $r_2,-r_1,r_3$\\
 \hline
 $SHSZ$ & $-r_1,r_3,r_2$\\
 $HX$ & $-r_3,r_2,r_1$\\
 $S$ & $-r_2,r_1,r_3$\\
 \hline
 $SHSXZ$ & $-r_1,-r_3,-r_2$\\
 $HY$ & $-r_3,-r_2,-r_1$\\
 $XS$ & $-r_2,-r_1,-r_3$\\
 \hline
\end{tabular}
 \captionof{table}{Possible transformations of a Bloch vector using unitary CSPOs.}
 \label{Bloch_vector_transformations}
\end{minipage}

\section{Geometrical interpretation of theorem~\ref{thm:qubit_interconversion}}\label{append:geometrical_interpretation}

To find whether a qubit can be converted to another using CSPOs, from Eq.~\eqref{eq:convex_comb} we get that it is enough to check whether the target state (or any of its Clifford equivalent state) lies outside the facets of the convex polytope (generated by the original state) that together cover any subset of any octant. 
For convenience, let us choose this subset to be the positive octant $(+X,+Y,+Z)$ for which $|\langle X \rangle|\leq |\langle Y\rangle|,|\langle Z\rangle|$ and denote it by $P_X$.
Hence, it is enough to find only those extreme points of the convex polytope which are used to form the facets that together cover $P_X$.
Using the hyperplane separation theorem, we can then find whether the target state lies inside this convex polytope.
Now, let the Bloch vector corresponding to $\rho$ (or its Clifford equivalent state) belonging to $P_X$ be denoted by $\mathbf{r_1} = (r_x, r_y, r_z)$.
We denote the neighbouring Clifford equivalent states which are used to form the facets of the convex polytopes as
\ba 
\mathbf{r_2} &= (r_z, r_x, r_y)\, ,\\
\mathbf{r_3} &= (r_y, r_z, r_x)\, ,\\
\mathbf{r_4} &=(-r_x, r_z, r_y)\, ,\\
\mathbf{r_5} &=(-r_y, r_x, r_z)\, ,\\
\mathbf{r_6} &= (r_y,-r_x, r_z)\, ,\\
\mathbf{r_7} &= (0, 0, 1)\, ,\\
\mathbf{r_8} &= (0, 1, 0)\, ,\\
\mathbf{r_9} &=(-r_z, r_y, r_x)\, ,\\
\mathbf{r_{10}} &= (r_z, r_y,-r_x)\, .
\ea
Now depending on the location of $\mathbf{r_1}$ in $P_X$, there are three possible ways to form a convex polytope.
Since we are only interested in the facets of these polytopes that cover $P_X$, we list below the set of vectors which, for each possible polytope, form a facet partially covering $P_X$ :

\textbf{Possibility 1:} $(\mathbf{r_1},\mathbf{r_6},\mathbf{r_7}),
						(\mathbf{r_1},\mathbf{r_7},\mathbf{r_5}),
						(\mathbf{r_1},\mathbf{r_5},\mathbf{r_4}),
						(\mathbf{r_1},\mathbf{r_4},\mathbf{r_3}),
						(\mathbf{r_1},\mathbf{r_3},\mathbf{r_2}),
						(\mathbf{r_1},\mathbf{r_2},\mathbf{r_6}),
						(\mathbf{r_3},\mathbf{r_4},\mathbf{r_8})$
						
\textbf{Possibility 2:} $(\mathbf{r_1},\mathbf{r_3},\mathbf{r_2}),
						(\mathbf{r_1},\mathbf{r_2},\mathbf{r_7}),
						(\mathbf{r_1},\mathbf{r_7},\mathbf{r_4}),
						(\mathbf{r_1},\mathbf{r_4},\mathbf{r_8}),
						(\mathbf{r_1},\mathbf{r_8},\mathbf{r_3})$
						
\textbf{Possibility 3:} $(\mathbf{r_1},\mathbf{r_{10}},\mathbf{r_3}),
						(\mathbf{r_1},\mathbf{r_3},\mathbf{r_2}),
						(\mathbf{r_1},\mathbf{r_2},\mathbf{r_4}),
						(\mathbf{r_1},\mathbf{r_4},\mathbf{r_9}),
						(\mathbf{r_1},\mathbf{r_9},\mathbf{r_8}),
						(\mathbf{r_1},\mathbf{r_8},\mathbf{r_{10}}),
						(\mathbf{r_4},\mathbf{r_2},\mathbf{r_7})$\\
In figures \ref{fig:poss1} and \ref{fig:poss2}, we have marked the location of the points in possibility 1 and possibility 2, respectively, highlighted (with red arcs) the subset they belong to, and connected the points in the way they are connected in the convex polytope for  a particular possibility.
\begin{figure}
\centering
\begin{minipage}{.5\textwidth}
  \centering
  \includegraphics[width=.8\linewidth]{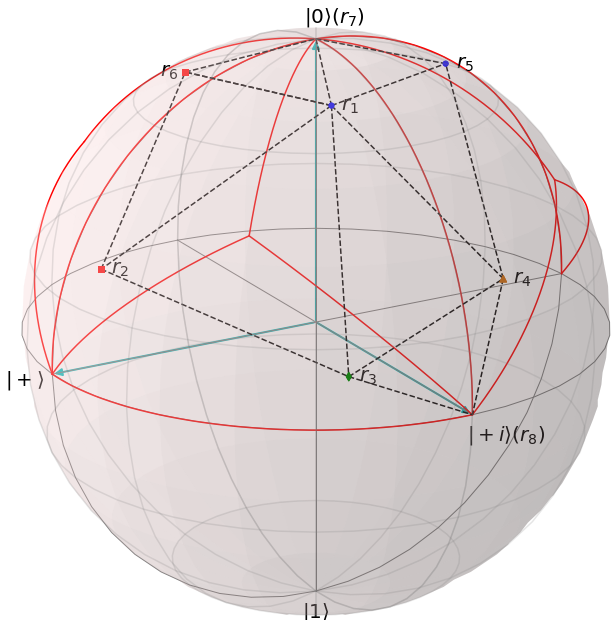}
  \captionof{figure}{Points corresponding to possibility 1}
  \label{fig:poss1}
\end{minipage}%
\begin{minipage}{.5\textwidth}
  \centering
  \includegraphics[width=.8\linewidth]{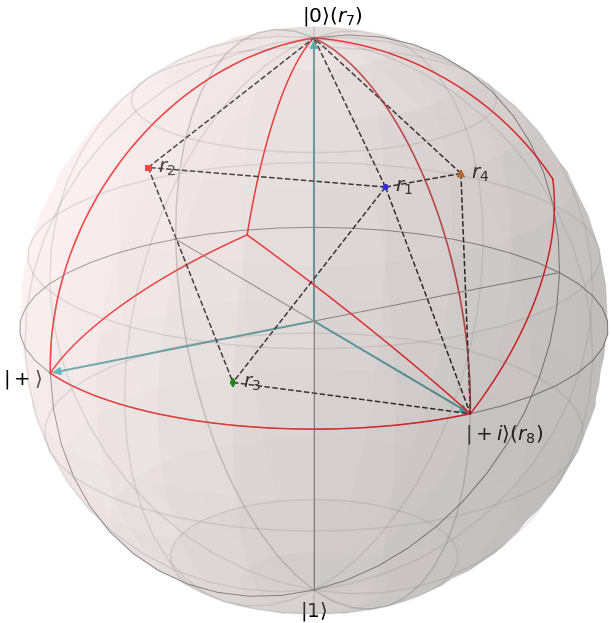}
  \captionof{figure}{Points corresponding to possibility 2}
  \label{fig:poss2}
\end{minipage}
\end{figure}
						
Using these set of vectors for each possible convex polytope, it is straightforward to find the vector (say $\mathbf{v}$) perpendicular to each facet such that the inner product of $\mathbf{v}$ with all vectors lying inside that facet is less than or equal to the inner product of $\mathbf{v}$ with one of the vectors on the surface of the facet. Let's call this inner product as $v$.
All the vectors on the other side of this facet will then give a value more than $v$ when their inner product is calculated with $\mathbf{v}$.
Therefore, by finding all such vectors perpendicular to each facet, we find the conditions to verify whether a vector lies inside or outside the facets.
Hence, a state $\rho$ can be converted to a state $\sigma$ using completely stabilizer preserving operations if and only if
\ba
\mathbf{s}\cdot \mathbf{u_i} &\leq u_i,\; \forall \, i ={1,\ldots, 7}\quad{\rm or}\\
\mathbf{s}\cdot \mathbf{v_j} &\leq v_j,\; \forall\, j ={1,\ldots, 5}\quad{\rm or}\\
\mathbf{s}\cdot \mathbf{w_k} &\leq w_k,\; \forall\, k ={1,\ldots, 7}
\ea
where $\mathbf{s}$ corresponds to the Bloch vector of the Clifford equivalent state of $\sigma$ in $P_X$. The vectors $\mathbf{u_i}$'s, $\mathbf{v_j}$'s, and $\mathbf{w_k}$'s are the vectors perpendicular to the facets of the respective possible polytopes, and $u_i$'s, $v_j$'s, and $w_k$'s are the constants which can be calculated from the inner product of $\mathbf{u}_i$, $\mathbf{v}_j$, and $\mathbf{w}_k$ with any vector lying on the surfaces of the respective facets of the possible polytopes.
\begin{remark}
The code for the above interconversion has been uploaded in a public git repository and can be freely accessed using the link in \footnote{https://github.com/gaurav-iiser/Resource-Theory-of-multiqubit-magic-channels}.
In the same link, we have also provided a code to construct a convex polytope from a given state.
The code can also be used to construct convex polytopes for various states at the same time, and hence can be used to check whether a convex polytope corresponding to some state lies inside another convex polytope or not.
\end{remark}

\end{appendix}

\end{document}